\documentclass[10pt]{iopart}

\usepackage{iopams}

\usepackage{bm}
\usepackage{graphicx}
\usepackage{dsfont}

\usepackage{bm}
\usepackage{mathrsfs}

\usepackage{framed}

\usepackage[all]{xy}
\usepackage{framed}
\usepackage{comment}
\usepackage{here}                    
\usepackage{latexsym}		     
\usepackage{amsfonts}  
\usepackage{amssymb}
\usepackage{amsthm}
\usepackage{dcolumn}
\usepackage{color}
\usepackage{mathrsfs}

\newtheorem{proposition}{Proposition}
\newtheorem{definition}{Definition}
\newtheorem{remark}{Remark}

\renewcommand{\phi}{\varphi}
\renewcommand{\>}{\rangle}
\newcommand{\<}{\langle}

\newcommand{\ket}[1]{|#1\>}
\newcommand{\bra}[1]{\<#1|}

\newcommand{\be}{\begin{equation}}
\newcommand{\ee}{\end{equation}}
\newcommand{\bea}{\begin{eqnarray}}
\newcommand{\eea}{\end{eqnarray}}

\newcommand{\Int}{\mathbb{Z}}

\renewcommand{\phi}{\varphi}

\begin{document}

\title[Computational indistinguishability and  boson sampling]{Computational indistinguishability and  boson sampling 
\footnote{This paper is dedicated to Igor Jex on the occasion of his 60th birthday.}
}

\author{Georgios M. Nikolopoulos}

\address{Institute of Electronic Structure and Laser, Foundation for Research and Technology - Hellas (FORTH),  GR-70013 Heraklion, Greece}
\vspace{10pt}
\begin{indented}
\item[]\today
\end{indented}

\begin{abstract}
We introduce a  computational problem of distinguishing between the output of an ideal  coarse-grained boson sampler and the  output of a true random number generator, as a resource for cryptographic schemes, which are secure against computationally unbounded adversaries. Moreover, we define a cryptographic setting for the implementation of 
such schemes, including message encryption and authentication, as well as 
entity authentication. 
\end{abstract}

%
\vspace{2pc}

%
%
%

\section{Introduction}

Computational indistinguishability plays a fundamental role in cryptography \cite{GoldreichBook}. Two probability distributions are considered to be computationally indistinguishable, if no efficient procedure can tell them apart. In particular, consider 
an efficient algorithm, which takes as input a string that has been drawn from one of the distributions under consideration. If the probability for the algorithm to accept a string that has been drawn from  one of the distributions 
is very close to the probability for  the same algorithm to accept a string taken from the other distribution, then 
the two distributions are considered to be indistinguishable. As far as  practical applications is concerned, 
objects that are computationally indistinguishable can be considered equivalent, because practical applications 
involve efficient algorithms and they cannot distinguish these objects.

Our main task here is to introduce a computational indistinguishability problem based on ideal coarse-grained boson sampling. Boson sampling (BS) is a computational problem, which is commonly believed to be of particular relevance in the quest for attaining the milestone of quantum advantage  \cite{BSintro,BSreview,BSreview2}. 
In a photonic setting, the problem pertains to sampling from the probability distribution of $N$ photons (bosons in general), which have 
gone through a random passive linear interferometer of $M>N$ modes, 
described by an $M\times M$ randomly chosen unitary transformation $\hat{\cal U}$.   
Since the formalization of the problem by Aaronson and Arkhipov \cite{Aaronson2013}, there have been great efforts by experimentalists to extend the scale of standard BS  \cite{Oxford_2013,Vienna&Jena_2013,Roma&Milano_2013,Roma&Milano_Birthday_2013,Brisbane_2013,Roma&Milano_validation_2014,Bristol_Haar_random_2014,Bristol_Universal_2015,Vienna&Jena_2015,Roma&Milano_2016,Brisbane_QD_2017,USTC_time_bin_2017,USTC_QD_2017,USTC_Lossy_2018,Oxford2019,Collisionfree,USTC2019} as well as its variants \cite{Roma&Milano_Scattershot_2015,USTC_Scattershot_2018,Jex17,Jex19,Bristol_Gaussian_2019,timestamp}.
The primary objective of all of these implementations has  been the  demonstration of  advantage of quantum computers over their classical counterparts, leaving aside other possible applications of BS, such as the ones discussed in Refs. \cite{Vibronic,Spin,Graph1,Graph2,Graph3,Nikolopoulos2016,Nikolopoulos2019,Feng-etal20,Shi-etal22}. 

In coarse-grained BS,  the set of all possible boson configurations is partitioned into $d$ 
disjoint subsets (to be referred to hereafter as bins). The problem under consideration is associated with the quest for the most-probable bin $\kappa$, in the coarse-grained boson distribution defined by an 
input boson configuration ${\bm s}$ and the unitary $\hat {\cal U}$. 
It is shown that, without knowledge of the input configuration, the probability distribution of 
the most-probable bin (MPB) $p_{\rm BS}(\kappa)$ can be exponentially close to the uniform distribution $p_{\rm uni}:=1/d$,  and thus a computationally unbounded adversary cannot distinguish between the two. 
Subsequently, we propose a specific cryptographic setting involving a trusted key-distribution center, which exploits the aforementioned computational indistinguishability, and allows for the implementation of various cryptographic tasks  between two honest users, including secure communication, entity authentication, and data-origin authentication. 

The paper is organized as follows. In section \ref{sec2} we give the necessary formalism for coarse-grained boson sampling, while section \ref{sec3} is devoted to computational indistinguishability. In section \ref{sec4} we define a 
cryptographic setting for applications, and a summary with concluding remarks is given in section \ref{sec5}.

\section{Formalism}
\label{sec2}

We begin with a brief formulation of ideal coarse-grained BS, which follows closely the one in  Refs. \cite{Nikolopoulos2016,Nikolopoulos2019}, and the interested reader may refer to these papers (and to the references therein) for more information. 
The problem of BS can be analyzed  in the framework of 
\bea
 |{\mathbb S}| = \left(
      \begin{array}{c}
        M+N-1 \\
        N  
      \end{array} \right) \sim \left (\frac{M}{N}\right )^N 
\label{H-size2}
\eea  
boson configurations $\{{\bf s}^{(0)}, {\bf s}^{(1)},\ldots\} := {\mathbb S}$.  The asymptotic value of $|{\mathbb S}|$ holds for 
$M\gg N$ and reflects the exponential growth of the space with the parameters $\{M,N\}$. 
The $M\times M$ unitary transformation $\hat{\cal U}$ is chosen randomly according to the Haar measure, and relates   
the bosonic annihilation/creation operators for the output modes to  
the corresponding operators for the input modes. 
There is a natural homomorphism between $\hat{\cal U}$ and 
the unitary operator $\hat{\cal V}$, which acts on $N$-boson quantum states, and thus one can use the two representations interchangeably. 
The $i$th input(output) boson configuration ${\bf s}^{(i)}\in{\mathbb S}$, is a tuple of $M$ distinct positive integers i.e., ${\bf s}^{(i)}:=(s_{1}^{(i)},\ldots,s_{M}^{(i)})$, where $s_{j}^{(i)}\in[1,N]$ refers to the number of bosons at the $j$th  mode.  
The set ${\mathbb S}$ is fully determined by the parameters $\{M,N\}$. 

Given an input configuration ${\bm s}\in{\mathbb S}$ (to be referred to hereafter as seed), the probability for obtaining 
configuration ${\bm r}$ at the output, 
is given by $Q({\bm r}|{\bm s};\hat{\cal U})\propto |{\rm Per}(\hat{\cal U}_{{\bm s},{\bm r}})|$, 
where ${\rm Per}(\hat{\cal U}_{{\bm s},{\bm r}})$ is the permanent of an $N\times N$ submatrix 
of $\hat{\cal U}$, which  is determined by the occupied input and output modes. 
The set of all possible boson configurations at the output ${\bm r} \in {\mathbb S}$,  together with the corresponding 
probabilities, define the distribution series ${\mathscr Q}({\bm s};\hat{\cal U}) := \{Q({\bm r}|{\bm s};\hat{\cal U})~|~{\bm r} \in {\mathbb S} \}$, which describes fully the problem of boson sampling for the given seed and unitary. 
Changing the seed and/or the unitary one has a different  distribution series, and thus a different sampling problem.

The set of boson configurations can be partitioned equally into $d$ disjoint subsets i.e., 
${\mathbb S} =  \bigcup_{l=0}^{d-1} {\mathbb B}_l$, so that 
all the bins have nearly the same size, and they may differ by at most one configuration, 
when the remainder of the 
division $R=|{\mathbb S}|/d\neq 0$. More precisely, the size of the $l$th bin is given by  
\bea
|{\mathbb B}_l|=\left \lfloor \frac{|{\mathbb S}|}{d}\right \rfloor + z_l,
\label{BinSize:eq}
\eea 
where $z_l$ is 1 for $l<R$ and 0 otherwise. 
The number of bins $d$ is assumed to scale polynomially with $\{M,N\}$, and at the same time, 
it  is  sufficiently large so that bosonic-interference effects survive binning, and they are reflected in the coarse-grained  distribution ${\mathscr P}({\bm s};\hat{\cal U}):= \{P(b|{\bm s};\hat{\cal U})~|~b \in \Int_d \}$. This distribution series 
defines fully the coarse-grained boson sampling for the given seed, unitary, and binning scheme.

\begin{remark}
\label{remark1}
One may also label the possible bins in terms of $n$-bit strings, 
where $n=\lfloor \log_2(d)\rfloor +1$ is the number of bits required for the 
unambiguous identification of a bin. 
\end{remark}

Let $\kappa$ be the label of the MPB in the distribution series ${\mathscr P}({\bm s};\hat{\cal U})$. 
The corresponding probability is given by 
\[
 P(\kappa|{\bm s};{\cal U}) = \sum_{{\bm r}\in{\mathbb B_\kappa}}  Q({\bm r}|{\bm s};\hat{\cal U}),
\]
where the summation is over all the output configurations that constitute the bin $\mathbb B_\kappa$. 
It is important to keep in mind, that for fixed binning scheme and fixed unitary, the MPB $\kappa$ in the distribution ${\mathscr P}({\bm s};\hat{\cal U})$ depends on the input configuration ${\bm s}$.  
An honest user who knows the seed and can search for the MPB by sampling from the right coarse-grained distribution, 
has an advantage over a potential adversary who does not know the seed, and  the coarse-grained distribution. 
This advantage has been used as a basis for the design of a one-way  function, which relies on coarse-grained BS \cite{Nikolopoulos2016,Nikolopoulos2019}. 

At this stage we have defined the general theoretical framework. A summary of the involved parameters, together with the underlying assumptions, are given in table \ref{tab1}. In the following section we discuss a computational distinguishability problem, 
which relies on the quest of the MPB in a coarse-grained BS distribution.

\begin{table}\caption{Summary of parameters and main assumptions or constraints.} \label{tab1}
\begin{tabular}{lcccc}
\br
Parameter & Public  &  Private & Assumptions/Constraints \\
\mr
Number of photons $N$ &  yes & no & $N\gg 1$\\ 
Number of modes $M$ &   yes & no & $M\gg N$\\ 
Set of boson configurations ${\mathbb S}$ & yes & no & $|{\mathbb S}|\sim (M/N)^N$\\
Number of bins $d$ & yes & no & $d\sim {\rm poly} (M,N)$\\
Unitary $\hat{\cal U}$ &  yes & no & Haar random\\ 
Seed ${\bm s}$ &  no & yes & ${\bm s}\in {\mathbb S}$\\
Most-probable bin (MPB) $\kappa$ & maybe & maybe &    $\kappa\in \Int_d$ \\
BS distribution ${\mathscr Q}({\bm r}|{\bm s};\hat{\cal U})$ &  no & yes & ${\bm r}\in {\mathbb S}$\\
Coarse-grained BS distribution ${\mathscr P}(b|{\bm s};\hat{\cal U})$ &  no & yes & 
$b\in \Int_d$ \\ 
\br
\end{tabular}
\end{table}

\normalsize

\section{Computational indistinguishability based on boson sampling}
\label{sec3}

As shown in table \ref{tab1}, throughout this work  all of the relevant parameters (including the binning scheme) are publicly known, apart from the seed ${\bm s}$. A BS-based distinguishability problem ({\tt BSDP}) can be defined as follows.

\begin{definition} {\rm ({\tt BSDP}):} For fixed and known $\{N, M, \hat{\cal U}\}$, and a given integer $\kappa \in\Int_d$, a verifier  has to decide which of the following is true: 
\\
(a) $\kappa$ has been chosen at random from a uniform distribution over $\Int_d$.\\
(b) $ \kappa$ refers to the MPB of a coarse-grained BS distribution ${\mathscr P}({\bm s};\hat{\cal U})$, 
for some ${\bm s}\in {\mathbb S}$. 
\end{definition}

For either of the two possible scenarios, $\kappa$ is a random integer. In scenario (a) it has a uniform distribution over 
$\Int_d$  i.e., the corresponding probability distribution is 
\bea
p_{\rm uni}(\kappa) = \frac{1}{d}. 
\label{p_uni:eq}
\eea
In scenario (b),   $\kappa$ has been obtained by  binning the data of a BS session 
with an unknown randomly chosen seed ${\bm s}$. 
From the verifier's point of view, who does not have access to ${\bm s}$, 
the corresponding probability distribution is 
\bea
p_{{\rm BS}}(\kappa)  = \sum_{{\bm s}\in {\mathbb S}} p({\bm s})P(\kappa|{\bm s};\hat{\cal U}),
\label{p_bs:eq}
\eea 
where $p({\bm s})$ is the probability distribution for the random seed ${\bm s}$. 
In this equation we have used the chain rule as well as the fact that the unitary is publicly known and fixed, while 
the seed is chosen at  random and independently. 
{\tt BSDP} essentially asks the verifier to distinguish between the 
probability distributions (\ref{p_uni:eq}) and (\ref{p_bs:eq}), 
based on the given integer ${\kappa}$, and all of the publicly known 
parameters $\{M,N,d;\hat{\cal U}\}$ \footnote{Note that this problem can be also viewed as  a decision problem \cite{Nikolopoulos2016}, because the verifier has to decide whether the  given  sequence $\kappa$, has been drawn from a uniform distribution or not. 
In this case, {\tt BSDP} accepts two possible answers namely, {\tt YES} or {\tt NO}. }.   
We will show the following proposition. 
\begin{proposition}
\label{prop1}
In the case of ideal BS, the distributions $p_{\rm BS}$ and  $p_{\rm uni}$ cannot be distinguished with probability better than an exponentially small number $\epsilon$, even by 
a computationally unbounded verifier. 
\end{proposition}

\begin{proof}
The distributions $p_{\rm BS}$ and  $p_{\rm uni}$ are over the same space $\Int_d$, and their 
statistical distance is defined as
\bea
{\cal D}:=\frac{1}2\sum_{\kappa\in \Int_d} \left | p_{\rm uni}(\kappa) - p_{\rm BS}(\kappa)\right |.
\label{stat_dist:eq}
\eea
It is known that when the distance between two 
distributions over the same space is upper bounded by $\epsilon\ll 1$, a computationally unbounded verifier cannot distinguish them with probability better then 
$\epsilon$ \cite{GoldreichBook,CryptoDistance}. So, to prove proposition \ref{prop1}, it suffices to show that ${\cal D} < \epsilon$, for some exponentially small $\epsilon\ll 1$.  

The conditional probability entering Eq. (\ref{p_bs:eq}), is the probability for the  
binned data of the BS session, to result in the MPB  with label $\kappa$ (i.e., the bin ${\mathbb B}_{\kappa}$). It can be related to the 
probabilities  entering the underlying 
BS distribution  ${\mathscr Q}({\bm r}|{\bm s};\hat{\cal U})$ as follows
 \bea
p_{\rm BS}(\kappa) 
&=& 
 \sum_{{\bm s}\in {\mathbb S}} p({\bm s})
 \sum_{{\bm r}\in{\mathbb B}_{\kappa}}
 Q({\bm r}|{\bm s};\hat{\cal U}) 
\nonumber \\
&=&
\sum_{{\bm r}\in{\mathbb B}_{\kappa}} \left [ 
 \sum_{{\bm s}\in {\mathbb S}} p({\bm s}) 
  Q({\bm r}|{\bm s};\hat{\cal U}) \right ]
 \nonumber \\
&=& \sum_{{\bm r}\in{\mathbb B}_{\kappa}} \left [ 
 \sum_{{\bm s}\in {\mathbb S}} p({\bm s}) 
  \left | \bra{\bm r}{\hat{\cal{V}}}\ket{\bm s}\right |^2 \right ].
\label{app2_uni2a}
\eea
In the last equation we have taken advantage of the 
natural  homomorphism between $\hat{\cal U}$ and $\hat{\cal V}$. 

If all the possible seeds are equally probable, 
the corresponding probability distribution is uniform i.e.,  $ p({\bm s})  = 1/|{\mathbb S}|$. 
Using the completeness relation we have 
\bea
p_{\rm BS}(\kappa) 
&=& \frac{1}{|{\mathbb S}|} 
\sum_{{\bm r}\in{\mathbb B}_{\kappa}} 
\bra{\bm r}
\hat{\cal V}
\left [ 
\sum_{{\bm s}\in{\mathbb S}}
\ket{\bm s}
\bra{\bm s}
\right ]
\hat{\cal V}^\dag
\ket{\bm r}
\nonumber\\
&=&\frac{1}{|{\mathbb S}|}  \sum_{{\bm r}\in{\mathbb B}_{\kappa}}  \bra{\bm r}
\hat{\cal V}
\hat{\mathbb I}
\hat{\cal V}^\dag
\ket{\bm r} =\frac{|{\mathbb B}_{\kappa}|}{|{\mathbb S}|},
\label{app2_uni1}
\eea
which shows that $p_{\rm BS}(\kappa)$ is  independent of $\hat{\cal U}$, and it is determined only by the size of the MPB relative to 
the size of the entire Hilbert space for $N$ bosons and $M$ modes.
Using Eqs. (\ref{app2_uni1}) and (\ref{BinSize:eq}) in Eq. (\ref{stat_dist:eq})  we have 
\bea
{\cal D}< \frac{d}{2|{\mathbb S}|}:=\epsilon
\label{epsilon:eq}
\eea
In view of Eq. (\ref{H-size2}), $|{\mathbb S}|$ scales exponentially with $(M,N)$, whereas   
$d\sim {\rm poly}(M,N)$. So, there is certainly a regime of values for $(M,N)$, where ${\cal D}$ is bounded by an 
 exponentially small quantity $\epsilon\ll 1$. In this case, 
the distributions $p_{\rm BS}$ and  $p_{\rm uni}$ cannot be distinguished with probability better than 
$\epsilon$, even by a computationally unbounded verifier. 
\end{proof}
  
Truly random dits can be generated only when we sample from the ideal uniform distribution $p_{\rm uni}$. The fact that $p_{\rm BS}$  is exponentially close to 
$p_{\rm uni}$, allows for the generation of random dits  (or equivalently random sequences of bits) 
by sampling from $p_{\rm BS}$. 
Typically, the performance of random number generators is quantified by the 
Shannon entropy  
\bea
{\mathscr H}_d(p):= -\sum_{x\in\Int_d} p(x)\log_d [p(x)]
\eea 
for some discrete distribution $p(x)$ on $\Int_d$.  
For an ideal uniform distribution one would have ${\mathscr H}_d(p_{\rm uni}) = 1\,{\rm dit}$.  
How close to ${\mathscr H}_d(p_{\rm uni}) $ is the Shannon entropy for random dits that 
are generated by sampling from $p_{\rm BS}$? 

\begin{proposition}
\label{prop2}
The Shannon entropy of distribution $p_{\rm BS}$ satisfies 
\bea
|{\mathscr H}_d(p_{BS}) - {\mathscr H}_d(p_{\rm uni}) | <  \frac{2d\log(d)}{\sqrt{|{\mathbb S}|}},
\label{H-dif_b:eq}
\eea
or equivalently  
\bea
|{\mathscr H}(p_{BS}) - {\mathscr H}(p_{\rm uni}) | <  \frac{2d}{\sqrt{|{\mathbb S}|} }.
\label{H-dif2_b:eq}
\eea
\end{proposition}

\begin{proof}
First of all, note that ${\mathscr H}_d(p) = {\mathscr H}(p) [\log(d)]^{-1}$, where everything on the rhs is for logarithm in base 2. Hence, we may consider  only one of the relations, let us say equation (\ref{H-dif2_b:eq}). 
In view of Proposition 1, we will focus on a regime where the statistical distance of the distributions 
$p_{\rm BS}$ and $p_{\rm uni}$ is upper bounded by $\epsilon\ll 1$. 
Assuming $\epsilon\leq 1/4$, we can use Lemma 2.7  in Ref. \cite{CK}  or Theorem 16.3.2 of Ref. \cite{EIT}  to obtain
\bea
|{\mathscr H}(p_{BS}) - {\mathscr H}(p_{\rm uni}) | \leq -2{\cal D} \log \left (\frac{2{\cal D}}{d} \right ).
\eea

For $d\gg 1$, the function on the r.h.s. of the inequality increases monotonically with ${\cal D}$ in $(0,1)$. 
Hence,  
\bea
|{\mathscr H}(p_{BS}) - {\mathscr H}(p_{\rm uni}) | < -2\epsilon \log \left (\frac{2\epsilon}{d} \right ), 
\eea
because ${\cal D} < \epsilon$.
Using the definition of $\epsilon$ [see Eq. (\ref{epsilon:eq})], we obtain 
\bea
|{\mathscr H}(p_{BS}) - {\mathscr H}(p_{\rm uni}) | < d \frac{\log(|{\mathbb S}|)}{|{\mathbb S}|}.
\label{prop2:ineq1}
\eea

Let us recall now that $|{\mathbb S}|\sim M^N \gg 1$ while 
$\log(|{\mathbb S}|)< \log(1+|{\mathbb S}|)$. 
Using the  inequality 
\bea
\ln(x+1)\leq \frac{x}{\sqrt{x+1}} \Rightarrow x^{-1}\log(x+1)\leq \frac{\log(e) }{\sqrt{x+1}}, 
\eea
we obtain 
\bea
\frac{\log(|{\mathbb S}|)}{|{\mathbb S}|}< \frac{2}{\sqrt{|{\mathbb S}|+1}} < \frac{2}{\sqrt{|{\mathbb S}|}} . 
\eea 
Equation (\ref{H-dif2_b:eq}) follows immediately from the last inequality and inequality (\ref{prop2:ineq1}).
\end{proof}

We have proved that by choosing $d$ to scale polynomially with $(M,N)$, and much slower than  $\sqrt{|{\mathbb S}|}$,  we have 
\[
|{\mathscr H}_d(p_{BS}) - {\mathscr H}_d(p_{\rm uni}) | \to 0.
\]
That is, as we increase the size of the Hilbert space, the entropy of  random dits that are generated from $p_{BS}$ converges to the entropy of truly random dits. 

\begin{remark}
\label{remark2}
For the sake of concreteness, all of the aforementioned discussion as well as the next section, have been given in the context of random integers. However, they can be also rephrased in terms of $n-$bit strings (see remark \ref{remark1}).  
\end{remark}

\section{A cryptographic setting with a trusted key-distribution center}
\label{sec4}

The computational indistinguishability discussed in the previous section, opens up the way for 
cryptographic applications. Consider for instance two honest users Alice (A) and Bob (B), and a trusted key-distribution center (KDC). 
The two users are connected to the KDC via  authenticated secure classical channels. Alice and Bob independently produce  
sequences of random dits 
\numparts 
\bea
\label{rand_dit:eqA}
&&\kappa_1^{(\rm A)}, \kappa_2^{(\rm A)}, \ldots, \kappa_j^{(\rm A)},\ldots\\
 &&\kappa_1^{(\rm B)}, \kappa_2^{(\rm B)}, \ldots,\kappa_j^{(\rm B)},\ldots,
 \label{rand_dit:eqB}
\eea
\endnumparts 
pertaining to independent Haar random, but publicly known, unitaries $\hat{\cal U}_{\rm A}$ and 
$\hat{\cal U}_{\rm B}$, respectively. In particular, the $j$th dits in the two sequences refer to the MPBs of the 
coarse-grained distributions ${\mathscr P}(b|{\bm s}_j^{\rm (A)};\hat{\cal U}_{\rm A})$ and  
${\mathscr P}(b|{\bm s}_j^{\rm (B)};\hat{\cal U}_{\rm B})$, where ${\bm s}_j^{\rm (A)}$ and ${\bm s}_j^{\rm (B)}$ 
have been chosen by Alice and Bob at random and independently, from a uniform distribution 
over ${\mathbb S}$. 
Alice and Bob keep track of their sequences of random seeds 
\numparts 
\bea
&&{\bm s}_1^{(\rm A)}, {\bm s}_2^{(\rm A)}, \ldots, {\bm s}_j^{(\rm A)},\ldots\\
 &&{\bm s}_1^{(\rm B)}, {\bm s}_2^{(\rm B)}, \ldots,{\bm s}_j^{(\rm B)},\ldots,
 \eea
\endnumparts 
while they send their dits [equations (\ref{rand_dit:eqA})-(\ref{rand_dit:eqB})]  to the KDC, where they are stored 
in a secure database in the form shown in table \ref{tab2}. The symbol $\oplus$ denotes addition modulo $d$.  

\Table{\label{tab2} The database of the trusted key-distribution center.} 
\br
Index & & Joint Key   \\
\mr
1 &   & $K_1 = \kappa_1^{(\rm A)}\oplus \kappa_1^{(\rm B)} $\\ 
2 & &  $K_2 =\kappa_2^{(\rm A)}\oplus \kappa_2^{(\rm B)}$ \\ 
$\vdots$ & & $\vdots$ \\
$j$ &  &$K_j=\kappa_j^{(\rm A)}\oplus \kappa_j^{(\rm B)}$ \\ 
$\vdots$ & & $\vdots$ \\
\br
\endTable

\normalsize

It is worth emphasizing here that the storage of the individual keys $\kappa_j^{\rm (A)}$ and $\kappa_j^{\rm (B)}$ is information-theoretically secure by virtue of the one-time pad encryption. 
Even if an adversary has access to the  database of the KDC, the extraction of the individual keys from the joint keys $K_j$ is impossible, and he cannot do better than random guessing. 
This is because each individual key has been chosen independently and at random from a distribution, which is computationally indistinguishable from uniform, according to the discussion in section \ref{sec3}. 
However, in order to ensure the security of the database, each entry has to be used only once i.e.,  
it has to be removed from the database immediately after its use. 

\subsection{Message encryption}
\label{sec4a}
Let us assume now that Alice wants to send a secret message $m\in\Int_d$ to Bob. 
To this end, the two users contact the KDC and they receive 
the joint key for this communication, say $K_j$, 
chosen at random from the list of available keys. As soon as the joint key is sent to the 
users, the corresponding entry is removed from the database, in order to ensure the one-time-pad character of the keys. 
Having received the joint key, Alice can calculate $\kappa_j^{\rm (A)}$ by sampling from   
${\mathscr P}(b|{\bm s}_j^{\rm (A)};\hat{\cal U}_{\rm A})$, because she knows the associated 
seed ${\bm s}_j^{\rm (A)}$ that has been used for its generation at first place. Similarly, Bob can extract $\kappa_j^{\rm (B)}$. 
Alice encrypts her message as follows
\bea
C_j = m\oplus K_j\oplus \kappa_j^{(\rm A)} =  m\oplus \kappa_j^{(\rm B)}. 
\eea
The ciphertext $C_j$ is sent to Bob, who can recover the message by  adding 
his key $\kappa_j^{\rm (B)}$ to the ciphertext. 
The particular communication scheme is information theoretically secure by virtue of the one-time pad encryption 
with the secret random key $ \kappa_j^{(\rm B)}$, which has a distribution  close to uniform.   

\subsection{Entity authentication}
\label{sec4b}
Another cryptographic task that can be implemented in the particular setting is entity authentication. 
Following the ideas discussed in Ref. \cite{EAP} Alice can obtain assurances about the identity of Bob, working as follows. She chooses at random a message $m$ from $\Int_d$, she encrypts it with $\kappa_j^{(\rm A)} $ 
and sends the ciphertext $C_j = m\oplus\kappa_j^{(\rm A)} $ to Bob, asking him to decrypt it and announce the output. 
If Bob is the one he claims to be, then he has access to ${\bm s}_j^{\rm (B)}$. So, 
he can calculate $\kappa_j^{\rm (B)}$ and through it he can calculate Alice's key 
$\kappa_j^{(\rm A)} = \kappa_j^{(\rm B)} \oplus K_j$. 
Hence, he can recover the random message that Alice has sent to him, and in this way Alice obtains assurance about Bob's true identity. 

\subsection{Data origin authentication}
\label{sec4c}
Let us assume that Alice and Bob share a common secret random key $k$ and Alice
 wants to send an authenticated message $m\in {\mathbb M}$ to Bob.  To this end, 
she evaluates the tag $t = h(k,m)$, 
where $h$ is a publicly known function, and let ${\mathbb T}$ denote the set of all possible different tags.  
The message and the tag $(m,t)$ are sent to Bob over a classical channel, and in general, as a result of forgery or noise, he 
will receive $(m^\prime,t^\prime)$. Using his key, Bob calculates $ h(k,m^\prime)$ and accepts the message only
if $t^\prime =h(k,m^\prime) $. 

It is well known that a $1$-time $2/|{\mathbb T}|$-secure MAC can be obtained by means 
of the Wegman-Carter construction, provided that the  common secret key $k$ is of length at least 
$4|\log_2({\mathbb T})|$ (see Ref. \cite{MAC} and references therein).  The construction relies on a strongly universal hash function $h_{k}$ \cite{book3,book4,WC81}, 
which is chosen at random from a class of such functions, and it is identified uniquely by the shared secret key $k$. 

In the cryptographic framework under consideration, Alice and Bob can establish a common secret key, 
by working as in the previous subsections. More precisely, they receive a randomly chosen joint key $K_j$ from the 
KDC. Subsequently,  Alice calculates $\kappa_j^{\rm (A)}$ and adds it to the joint key in order to recover  
Bob's key $\kappa_j^{\rm (B)} = K_j\oplus \kappa_j^{\rm (A)}$. Bob  can also extract 
$\kappa_j^{\rm (B)}$ by sampling from ${\mathscr P}(b|{\bm s}_j^{\rm (B)};\hat{\cal U}_{\rm B})$, 
because he knows ${\bm s}_j^{\rm (B)}$. 
Hence, in the particular example, the common secret key to be used for the authentication is $k=\kappa_j^{\rm (B)}$.

In closing this section, it is worth noting that in the scenario discussed in subsection \ref{sec4a}, the content of the message has to be protected from potential adversaries, and thus it is never announced  publicly. In the framework of entity authentication discussed in section \ref{sec4b}, the message is assumed to be randomly chosen from a uniform distribution over $\Int_d$, and serves as a means for the receiver to convince Alice about his identity in a single authentication session.  
Finally, in section \ref{sec4c}  the message is sent in clear together with the tag, 
because we are interested only in authenticating the origin of the message.

\section{Concluding remarks}
\label{sec5}

We have defined computational indistinguishability in the context of ideal coarse-grained boson sampling, and  have 
discussed the conditions for its validity under the assumption of a computationally unbounded verifier. 
Subsequently we defined a cryptographic 
framework, which involves a trusted key distribution center, and allows for the implementation of various cryptographic tasks  
between two legitimate users who posses independent boson samplers. 

Interestingly enough, the computational indistinguishability  we have discussed seems to be independent of whether 
the boson sampling is performed in a regime where the problem becomes intractable for classical computers. 
Mainly it relies on the unitarity of the matrix that connects the input modes to the output modes, as well as on the 
assumption that the verifier does not have access to the input configuration.   
So, the boson sampling required by the cryptographic tasks discussed in section \ref{sec4} can be performed on  a classical 
computer or a quantum chip. In the former case, it is easy to ensure the unitarity of the matrix, whereas in the 
latter case, as a result of losses, boson sampling essentially involves a linear non-unitary complex-valued 
transformation. An interesting question therefore, which goes beyond the scope of the present work, 
is whether and how the present results can be generalized to lossy and noisy boson samplers. 

In closing, it should be emphasized that the main aim of the present work was 
to discuss computational indistinguishabilily in the framework of boson sampling. To the best of our knowledge, this is the first attempt in the literature for the boson sampling to be 
connected to some sort of computational indistinguishabilily. The cryptographic setting discussed in section \ref{sec4}, is far from being practical, and cannot in any case compete with existing  cryptographic algorithms. It has to be viewed as a toy model, which demonstrates how boson sampling can, in principle, 
be used for specific cryptographic tasks. If the  computational indistinguishabilily  can be extended to lossy and noisy on-chip boson samplers, then the present toy model, with the necessary amendments,  may start getting more practical.

\end{document}